\pdfoutput=1
\documentclass[11pt, a4paper]{article}
\usepackage[english]{babel}
\usepackage[latin1]{inputenc}
\usepackage[T1]{fontenc}
\usepackage{epic,eepic}
\usepackage{setspace}
\usepackage[dvips]{graphics}
\usepackage{epsfig}
\usepackage{lscape}
\usepackage{rotating}
\usepackage{amsmath}
\usepackage{amssymb}
\usepackage{amscd}
\usepackage{fancyhdr}
\usepackage{tabularx}
\usepackage{color}
\usepackage{indentfirst}
\usepackage{floatflt} 
\usepackage{slashbox}
\usepackage{multicol}
\newtheorem{lemma}{Lemma}[section]

\newenvironment{proof} [1] {\begin{trivlist} 
  \item [\hskip \labelsep {\bfseries Proof}]~#1} {\end {trivlist}}

\begin{document}
\nocite{*}
\bibliographystyle{plain}
\thispagestyle{plain}
\begin{center}
  {\Large {\bf General Iteration graphs and Boolean automata circuits}}

  \vspace{1cm}
  {\large Mathilde Noual$^{1,2}$}\\

  \vspace{1cm}
  {\small
  \begin{itemize}
    \item[$^2$] Universit\'e de Lyon, \'ENS-Lyon, LIP, CNRS UMR5668, 69007 Lyon,
    France\vspace{-5pt}
    \item[$^4$] IXXI, Institut rh\^one-alpin des syst\`emes complexes, 69007 Lyon,
    France\vspace{-5pt}
  \end{itemize}}
\end{center}
\centerline{\rule[0pt]{4cm}{0.3pt}}
\vspace{1cm}

\begin{abstract}
  \noindent This article is set in the field of regulation networks modeled by
  discrete dynamical systems. It focuses on {\em Boolean automata networks}. In
  such networks, there are many ways to update the states of every element. When
  this is done deterministically, at each time step of a discretised time flow
  and according to a predefined order, we say that the network is updated
  according to a block-sequential update schedule (blocks of elements are
  updated sequentially while, within each block, the elements are updated
  synchronously). Many studies, for the sake of simplicity and with some
  biologically motivated reasons, have concentrated on networks updated with one
  particular block-sequential update schedule (more often the
  synchronous/parallel update schedule or the sequential update schedules). The
  aim of this paper is to give an argument formally proven and inspired by
  biological considerations in favour of the fact that the choice of a
  particular update schedule does not matter so much in terms of the possible
  and {\em likely} dynamical behaviours that networks may display. 
\end{abstract}
\noindent {\em Keywords:} Discrete dynamical system, regulation network,
positive and negative circuit, asymptotic dynamical behaviour, attractor,
potential, update schedule.

\section{Introduction}
As many studies have already
emphasised~\cite{comp,Julio_rob,these_adrien,BScircuits}, amongst the features
of discrete models of regulation networks that impact significantly on their
dynamical behaviour are their update schedules. Generally speaking, an update
schedule specifies when the elements of the network are updated throughout the
flow of discretised time. Block-sequential update schedules are amongst the most
famous. Their characteristic is to update all elements deterministically and
exactly during a fixed amount of time. It is well known that different
block-sequential update schedules yield different dynamical behaviours of
networks. Thus, the choice of one specific block-sequential update schedule
cannot be justified reasonably solely by its definition itself.  From the
biological point of view, the lack of knowledge concerning the order of gene
regulations does not help either in giving an argument in favour of one
iteration mode rather than another. Biologists, however, tend to agree that the
probability that all genes involved in a same cellular physiological function
evolve synchronously is almost null, particularly considering the plausible
presence of noise. Furthermore, it does not seem reasonable to think that all
genes (and their expressions) are subjected to a particular genetic biological
clock and that the biological clocks of all genes are synchronised although a
total lack of synchronicity appears to be on the whole rather unlikely as well.
\medskip

In order to bypass the problem of the choice of update schedule, we suggest here
to show that many of the dynamical behaviours induced by particular
block-sequential update schedules are in fact meaningless artefacts in a sense
that we will clarify. On the contrary, certain sets of network configurations
are indeed stable enough to resist to likely perturbations of the update
order. Fixed points are well-known and simple examples of these stable
configurations. To point this out, we define general iteration graphs that
contain all the information concerning all possible deterministic dynamical
behaviours of a network.   
\medskip

Our study is carried out on the discrete models of regulation networks that are
Boolean automata networks. These networks are described in Section~\ref{def}. In
section~\ref{appl}, with the intent of proving the pertinence of general
iteration graphs, we analyse the general iteration graph of a real
network. Section~\ref{circ} focuses on Boolean automata circuits, that is
networks whose underlying structures are circuits.  Special attention is payed
to these particular networks in this document. The reason is that circuits, as
Thomas discovered in 1981~\cite{Thomas1981}, play an important role in the dynamics of
networks containing them.  One way to see this is to note that a network whose
underlying interaction graph is a tree or more generally a graph without
circuits can only eventually end up in a configuration that will never change
over time (a fixed point). A network with retroactive loops, on the contrary,
will exhibit more diverse dynamical behaviour patterns. 
Thomas~\cite{Thomas1981} formulated conjectures concerning the role of positive
({\it i.e.,} with an even number of inhibitions) and negative ({\it i.e.,} with
an odd number of inhibitions) circuits in the dynamics of regulation
networks. At the end of section~\ref{circ}, we discuss how our work agrees with
them.  

\section{Definitions, and notations}
\label{def}
\subsection*{Boolean automata networks}
We define a {\bf Boolean automata network} to be a couple $N=(G,{\cal F})$ where
$G=(V,A)$ is the {\bf interaction graph} of the network and ${\cal F}=\{f_i\ |\
i\in V\}$ is its {\bf set of local transition functions}.  The elements of the
network are represented by the nodes of $G$ (in the sequel we will sometimes
identify a network with its interaction graph). Each one of them has a Boolean
state that may change over time. It is either active (its state is $0$) or
inactive (its state is $1$).  If the size of $N$ is $n$, {\it i.e.,} if $|V|=n$,
vectors in $\{0,1\}^n$ are called {\bf configurations} or {\bf (global) states}
of $N$.  In the interaction graph $G$, an arc $(i,j)\in A$ indicates that the
state of the element or node $j\in V$ depends on that of the node $i\in V$. The
local transition function $f_j:\{0,1\}^n\to\{0,1\}\in {\cal F}$ indicates
how. More precisely, starting in a configuration $x\in \{0,1\}^n,$ if the state
$x_j$ of node $j$ is updated, it becomes
$$x'_j=f_j(x).$$
The local transition functions of a network are 
supposed to be monotonous and if $(i,j)\in A$, then, when
\begin{multline*}
  \forall (x_0,\ldots,x_{n-1})\in\{0,1\}^n,\\
  f_j(x_0,\ldots,x_i=1,\ldots,x_{n-1}) \geq
  f_j(x_0,\ldots,x_i=0,\ldots,x_{n-1}),
\end{multline*}
node $i$ is said to be an {\bf activator} of node $j$ and the arc $(i,j)\in A$
is said to be {\bf positive} (it is labeled by $\oplus$). Otherwise, $i$ is said
to be an {\bf inhibitor} of $j$ and the arc $(i,j)\in A$ is said to be {\bf
  negative} (it is labeled by $\ominus$).  As the function $f_j$ only depends on
the coefficients $x_i$ of $x$ that correspond to states of incoming neighbours
$i$ of node $j$, we consider it to be a function of
$\{0,1\}^{|V_j^-|}\to\{0,1\}$, where $V_j^-=\{i\ |\ (i,j)\in A\}$, and we write
the following:
\begin{equation}
\label{update}
x_j'=f_j(x_i\ |\ i\in V_j^-). 
\end{equation}

\subsection*{Block-sequential update schedules}
The elements and the interactions of a network being specified, one last point
needs to be clarified in order to define completely a network and in particular
the way it evolves over time, that is, {\em when} is the state of each node $j$
updated?  \medskip

Networks are often associated to an {\bf update} or {\bf iteration schedule}
that specifies the order according to which the nodes are updated. One of the
most common deterministic update schedules are {\bf block-sequential update
  schedules}. These update schedules can be defined formally by a function $s:
V\to \{0,\ldots, |V|-1\}$. Then, $s(i)$ represents the date of update of node
$i$ {\em within} one unitary time step ({\it i.e.,} between a time step $t$ and
the following time step $t+1$). At the end of each time step, all nodes of the
network have been updated exactly once since the previous time step. Without
loss of generality, we suppose that $s$ allows for no ``waiting period'' within
a time step: $min \{s(i)\ |\ i\in V\}= 0$ and $\forall t,\ 0\leq t <n-1,\
\exists i\in V,\ s(i)=t+1\ \Rightarrow\ \exists j\in V,\ s(j)=t$.  The update
schedule called {\bf synchronous} or {\bf parallel} is denoted by $\pi$. It is
such that $\forall i\in V,\ \pi(i)=0$.  An update schedule $s$ is said to be
{\bf sequential} when it updates only one node at a time: $\forall i,j\in V$,
$s(i)\neq s(j)$, {\it i.e.,} $\forall t,\ 0\leq t<n,\exists i\in V,\
s(i)=t$. There are $n!$ different sequential update schedules of a set of $n$
nodes. Block-sequential update schedules take there name from the fact that they
define {\bf blocks} $B_k^s=\{i\in V\ |\ s(i)=k\}$ of nodes that are updated
synchronously while the blocks are updated sequentially.  From~\cite{adrien}, we
know that the number of different block-sequential update schedules of size $n$
is given by the following formula:
$${\cal B}(n)=\sum_{k=0}^{n-1} \binom{n}{k}{\cal B}(k)$$
(that is, the number of lists of sets of elements taken in a set of size $n$).
\medskip

As a finite sized network $N=(G,{\cal F})$ has a finite number of
configurations, with a deterministic (block-sequential) update mode $s$, it
necessarily ends up looping on a certain set of configurations ${\cal
  A}\subseteq \{0,1\}^n$ (where $n$ is still the network size). All such sets
${\cal A}$ of configurations that a network may reach after a certain number of
steps and that it can then never leave are called the {\bf attractors} of the
network updated with th specific update schedule $s$. For any initial
configuration $x(0)\in\{0,1\}^n$, let $x(t)$ be the configuration of the network
after $t$ steps of time during which the updates of nodes have been performed
according to the local transition functions of ${\cal F}$ and to the update
schedule $s$. Then, an attractor of $N$ is a set of $p\in \mathbb{N}$
configurations $x(t), x(t+1),\ldots,x(t+p-1)$, such that $\forall t'>0,\
x(t+t')=x(t+t'\ mod\ p)$. If, in addition, there is no $d<p$ satisfying $\forall
t',\ x(t'+ d)=x(t')$, then the attractor is said to have a period equal to $p$.
Attractors of period $1$ are usually called {\bf fixed points} and other
attractors (attractors of period greater than $1$) are called {\bf limit
  cycles}.\medskip

Now, the formalism described above implies that within one unitary time step,
the network undergoes in reality several update stages, one for the update of
each block $B_k^s=\{i\in V\ |\ s(i)=k\}$ (except when the update schedule is the
parallel schedule which defines only one block). Observing this way the network
only at regular intervals corresponding to one unitary time step rather than
observing it each time an update occurs makes sense since during each of these
intervals, the sequence of events that occurs is always the same. However,
although the choice of the moment of observation (after the last block $B_k^s$
has been updated) is {\it a priori} arbitrary, its consequence is to distinguish
update schedules that could reasonably be identified. Indeed, let $s$ and $s'$
be two update functions of $V\to\{0,\ldots,n\}$ such that $s_{max}=max\{s(i)\}$
and $\exists \Delta\in \mathbb{N},\ \forall i\in B^s_k,\ s'(i)=k+\Delta\ mod\
s_{max}$. Obviously, the formal definition of block-sequential update schedules
we gave above distinguishes between $s$ and $s'$ and indeed $s$ and $s'$, seen
as block-sequential update schedules, may in appearance induce very dissimilar
dynamical behaviours of the network: the state of the network just after the
nodes of $B^s_{s_{max}}$ are updated may often be different from the state of
the network just after the nodes of $B_{s_{max}-\Delta}$ are updated.  Yet, when
applied repeatedly, $s$ from $s'$ only differ by the very first element they
update. If we identify all such update schedules, the number of different
block-sequential update schedules becomes the following:
$$
{\cal B}'(n)=\sum_{k=0}^{n}\frac{1}{k}\cdot {\cal S}(n,k)
$$
where ${\cal S}(n,k)=k\cdot ({\cal S}(n-1,k-1)+{\cal S}(n-1,k))$ is the number
of lists of $k$ sets of elements taken in a set of size $n$, {\it i.e.,} the
number of surjections of $\{1,\ldots,n\}\to\{1,\ldots,k\}$.  
\medskip

With this new definition, a network updated with a block-sequential update
schedule needs to be observed after each update of a subset of its
nodes. Keeping this way of observing a network, we generalise the notion of
updating by loosing the condition on the order with which nodes are updated. For
any subset $P\subseteq V$ we define the {\bf global transition
  function relative to $P$}, $F^P:\{0,1\}^n\to\{0,1\}^n$:
$$
\forall x\in \{0,1\}^n,\ \forall i\in V,\ F^P(x)_i=\begin{cases}x_i & \text{ if
  }i\notin P\\ f_i(x_{i-1}) & \text{ if }i\in P.\end{cases}
$$
When $P=V$, we write $F=F^P$ and $F^k=F\circ F^{k-1}$ ($F^1$ being $F$). $F$ is
the global transition function of a network with $n$ elements updated with the
parallel update schedule. In the sequel, we make substantial use of the
following notation:
$$
\forall x\in \{0,1\}^n,\ U(x)=\{i\in V\ |\ x_i\neq f_i(x_{i-1})\} \text{ and }
u(x)=|U(x)|.
$$
In other words, when the network is in configuration $x$, $U(x)$ (which was
called {\em call for updating} and noted $Updx(x)$ in~\cite{Remy}) refers to the
set of nodes that are {\it unstable} in the sense that their states {\it can}
change (but do not necessarily, this depends on the subset $P$).  Note that by
definition of $P$ and $U(x)$:
$$
F^P(x)_i=
\begin{cases}
  f_i(x_{i-1}) & \text{if } i\in P\cap U(x)\\
  x_i &  \text{otherwise.}
\end{cases}
$$

The {\bf general iteration graph} of a network $N$ is a digraph ${\cal D}(N)$
whose nodes are the configurations $x\in\{0,1\}^n$ of $N$ and whose arcs are the
$(x,y)\in \{0,1\}^n\times \{0,1\}^n$ such that $\exists P\subseteq V,\ p\neq
\emptyset\ y=F^P(x)$.  Each node of this graph has thus an outdegree equal to
the power set of $V$ minus one (for the empty set), that is
$2^n-1$. Figures~\ref{it} and~\ref{itpos} picture the general iteration graphs
of two particular networks, namely a negative and a positive Boolean automata
circuit of size $3$ (Boolean automata circuits are defined formally in the next
section). \medskip
\begin{figure}[htbp!]
  \begin{center}
    \hspace{-3cm}\scalebox{0.6}{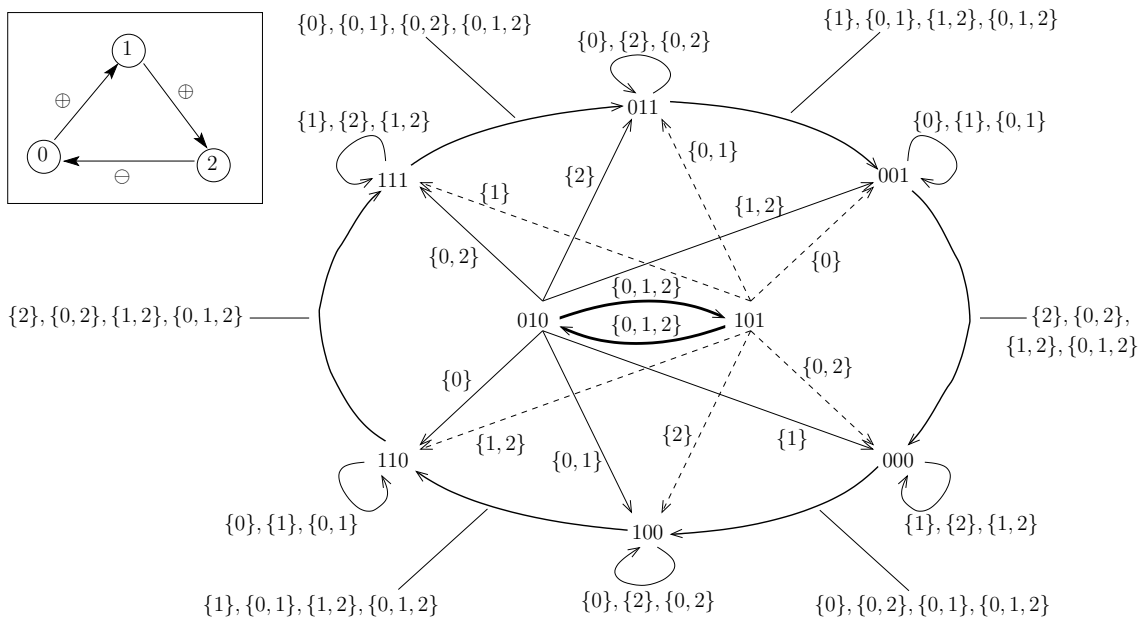}
    \caption{{\small General iteration graph of a negative Boolean automata
        circuit of size $3$. Arcs labeled by the empty set have been
        omitted. Arcs labeled by several subsets represent several arcs labeled
        by one subset. The different fonts of arcs have been chosen for the sake
        of clarity and not of meaning.  The interaction graph of the network is
        pictured in the frame.}}\label{it}
  \end{center}
\end{figure}
\begin{figure}[htbp!]
  \begin{center}
    \hspace{-3cm}\scalebox{0.6}{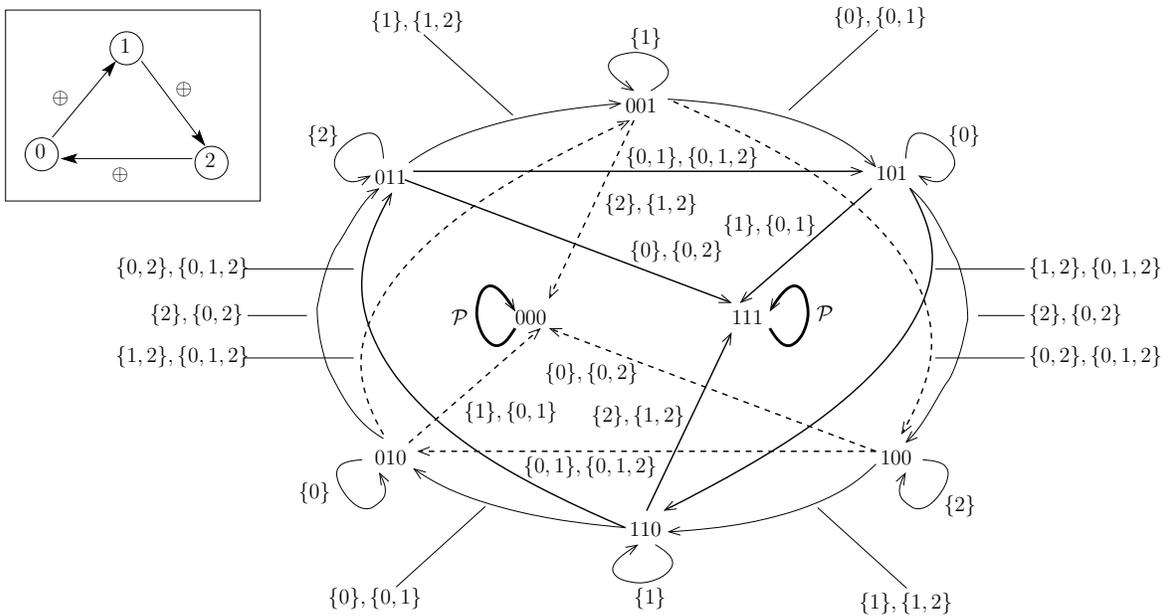}
    \caption{{\small General iteration graph of a positive Boolean automata circuit of
        size $3$. Arcs labeled by the empty set have been omitted. Arcs labeled
        by several subsets represent several arcs labeled by one subset. The label
        ${\cal P}$ is equivalent to the label $\{0\}, \{1\}, \{2\},
        \{0,1\},\{0,2\},\{1,2\},\{0,1,2\}$, {\it i.e.,} the power set of
        $\{0,1,2\}$ minus the empty set. The different fonts of arcs have been
        chosen for the sake of clarity and not of meaning. The interaction graph
        of the network is pictured in the frame. }}
  \end{center}
  \label{itpos}
\end{figure}

\section{General iteration graph of a real network model}
\label{appl}
In~\cite{theseS}, Sen\'e studies a regulation network of size $12$ modeling the
floral morphogenesis of plant \emph{Arabidopsis thaliana}. We will call this
network the {\it Mendoza network} for short because it was first introduced by
Mendoza and Alvarez-Buylla~\cite{mend}. With a sequential update schedule
this network has six fixed points.  With the parallel update schedule, Sen\'e
in~\cite{theseS} finds that it has the same six fixed points%
\footnote{It is well known and easy to see that all block-sequential updates of
  a network induce the same fixed points.} but also seven limit cycles of period
$2$. Whereas the fixed points have some biological meaning in that they
correspond to effective or hypothetical cellular types, the limit cycles are
believed to have none {\it a priori}. In this section, we use general iteration
graphs to suggest an explanation of this difference observed between the
behaviour of the model and the known behaviour of the real network.  \medskip

In~\cite{theseS}, a simplified version of the Mendoza network is presented. Let
us denote it by $N$. The network $N$ has only two strongly connected components%
. The rest of the network contains no cycles so that after a certain number of
time steps, the states of all nodes not belonging to the two strongly connected
components become fixed. These components thus represent, in a sense, the
``motor of the network dynamics''. This is why, here, we focus on them and on
the few eventual nodes that act directly on them. We call $N_1=(G_1,{\cal F}_1)$
and $N_2=(G_2,{\cal F}_2)$ the two sub-networks corresponding to the two
strongly connected components of $N$ and their surrounding nodes. Their
interaction graphs $G_1$ and $G_2$ are represented in
figures~\ref{fig_mend}~$a.$ and~\ref{fig_mend}~$b.$, respectively. Let us notice
that as soon as node $3$ of $G_2$ is updated once, it takes the same state as
node $2$ and neither of these two nodes ever changes states again. Therefore, we
can make out two different cases. The first one is when node $2$ is initially in
state $0$ and the second one is when it is initially set to $1$. Assuming node
$3$ is updated at least once, to study the dynamics of $N_2$, we only consider
the eight configurations $x\in\{0,1\}^4$ in which $x_2=x_3$.
\begin{figure}[htbp!]
  \begin{center}
    \begin{tabular}{cc}
    \scalebox{0.7}{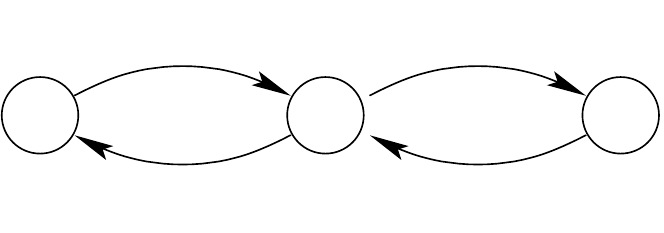}\hspace{1cm}& \hspace{1cm}
    \scalebox{0.7}{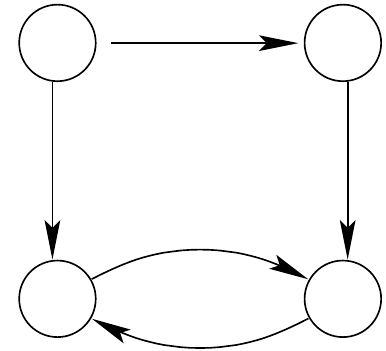}\\[0.4cm]
    $a.$ $G_1$ \hspace{1cm}& \hspace{1cm} $b.$ $G_2$\\[1cm]
    \scalebox{0.7}{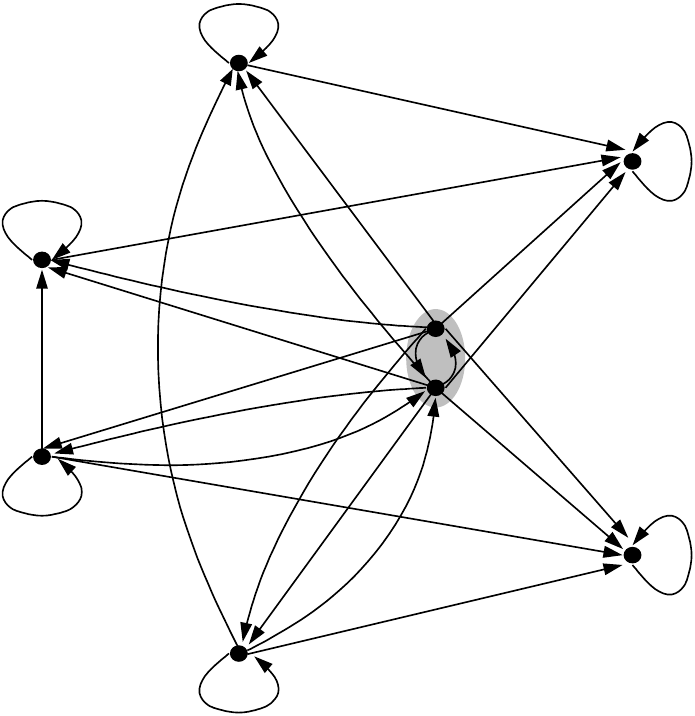}\hspace{1cm}&\hspace{1cm}
    \scalebox{0.7}{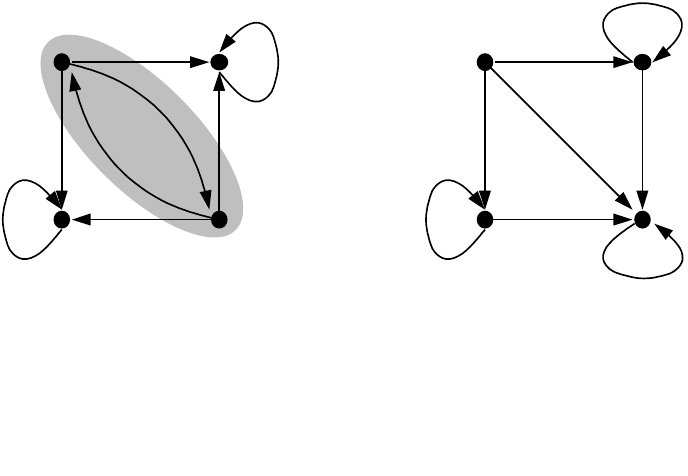}\\[0.4cm]
    $c.$ ${\cal D}(N_1)$ \hspace{1cm}& \hspace{1cm} $d.$ ${\cal D}(N_2)$\\
    \end{tabular}
    \caption{{\small Sub-networks $N_1=(G_1,{\cal F}_1)$ and $N_2=(G_2,{\cal
          F}_2)$ of the Mendoza network and their dynamics. Figures $a.$ and
        $b.$ picture the interaction graphs $G_1$ and $G_2$. Nodes $0,1,2$ of
        $G_1$ represent, respectively, genes {\sc AP3}, {\sc BFU}, {\sc PI}
        (see~\cite{theseS}). Nodes $0,1,2,3$ of $G_2$ represent, respectively,
        genes {\sc AP1}, {\sc AG}, {\sc EMF1}, {\sc TFL1}. Figures $c.$ and $d.$
        represent the general iteration graphs ${\cal D}(N_1)$ and ${\cal
          D}(N_2)$ of the sub-networks $N_1$ and $N_2$. In these graphs, for the
        sake of clarity, we have represented $\mathrm{n>1}$ arcs with the same
        beginning and ending as one unique arc labelled by $\mathrm{n}$. We have
        omitted the subset labels. In figure $d.$ there are two separate
        graphs. The one on the left corresponds to the case where nodes $2$ and
        $3$ of $G_2$ are both in state $0$ and the one on the right to the case
        where they are both in state $1$.  In all general iteration graphs of
        figures $c.$ and $d.$, shadowed subgraphs correspond to limit cycles of
        $N_1$ or $N_2$ that are observed with the parallel update schedule.}}
  \end{center}
  \label{fig_mend}
\end{figure}

In the general iteration graphs of $N_1$ and $N_2$ represented in
figures~\ref{fig_mend}~$c.$ and~\ref{fig_mend}~$d.$, appear the limit cycles
that are induced by the parallel update schedule (shadowed sub-graphs in
figures~\ref{fig_mend}~$c.$ and~\ref{fig_mend}~$d.$).  These limit cycles of the
sub-networks are responsible for the limit cycles of the network $N$ updated in
parallel. Now, let us suppose that although the network modeled by $N$ is indeed
updated in parallel, errors in the updating order do occur from time to time.
In order to analyse formally the chances that one of these limit cycles has to
be effectively observed, let us introduce some very coarse notions of {\em
  robustness} and {\em likeliness} of a set of configurations. Our aim here is
only to give a rough intuition on the matter of these unexplained limit
cycles. We certainly do not claim to introduce some subtle and indisputable
mathematical tools to analyse general iteration graphs. Thus, although the
following two functions will serve us now as measures%
\footnote{Note that, strictly speaking ${\cal R}$ and $\mathbb{P}$ are not
  mathematical measures.} of, respectively, the {\em robustness} and the {\em
  likeliness} of a set of configurations $C\subseteq\{0,1\}^n$, we believe that
these notions surely call for much finer definitions.
$$
{\cal R}(C)=\frac{1}{deg^{+}(C)}\hspace{2cm}
\mathbb{P}(C)=\frac{deg^{-}(C)}{{\cal T}_{\overline{C}}}
$$
Above, ${\cal T}_{\overline{C}}$ is the number of arcs $(x,y)$ in ${\cal D}(N)$
that are not between configurations of $C$ ($x,y\notin C$) and
$deg^+(C)=|\{(x,y)\ |\ x\in C,y\notin C \}|$.
For the set of configurations that are shadowed in figures~\ref{fig_mend}~$c.$
and~\ref{fig_mend}~$d.$, the value of ${\cal R}(\cdot)$ is smaller than $1$ and
the value of $\mathbb{P}(\cdot)$ is null. On the contrary, for all fixed points
$x$ in these figures, ${\cal R}(\{x\})=1/0=\infty$ is obviously maximal and
$\mathbb{P}(x)>0$.  In other words, fixed points appear to be ``robust'' and
``likely'' in the sense that many arcs lead to them and no arcs leave them.  As
for limit cycles, on the contrary, not only very few arcs of the general
iteration graphs lead to them but also, all of their outgoing arcs lead to
configurations that do not belong to them.  Thus, starting in one arbitrary
configuration, the network has very few chances to end in a configuration
belonging to a limit cycle and if ever it does, the chances are that it will
leave it very quickly. Assuming that it is doubtful that real networks such as
the one commanding the floral morphogenesis of \emph{Arabidopsis thaliana} may
obey infallibly the exact same (block-sequential) updating order of their
elements, this, we believe, may be interpreted as evidence of the fact that the
limit cycles of the Mendoza network observed with the parallel update schedule
are highly improbable to be reached and maintained over time, contrary to its
fixed points.

\section{Boolean automata circuits}
\label{circ}
A {\bf circuit} of size $n$ is a directed graph that we will denote by
$\mathbb{C}_n=(V,A)$. We will consider that its set of nodes,
$V=\{0,\ldots,n-1\}$, corresponds to the the set of elements of
$\mathbb{Z}/n\mathbb{Z}$ so that, considering two nodes $i$ and $j$, $i+j$
designates the node $i+j\mod n$. The circuits set of arcs is then $A=\{(i,i+1)\
|\ 0\leq i<n\}$. Let $id$ be the identity function ($\forall a\in\{0,1\},\
id(a)=a$) and $neg$ the negation function ($\forall a\in\{0,1\},\ neg(a)=\neg
a=1-a$). A {\bf Boolean automata circuit} of size $n$ is a 
network whose interaction graph is a circuit and whose set of local transition
functions is included in $\{id,neg\}$.
For this particular instance of a Boolean automata network, the update rule
given by equation~\ref{update} simplifies to:
\begin{equation}
  \label{update_circ}
  x_j'=f_j(x_{i-1}).
\end{equation}
With the restriction on the local transition functions, $f_i\in\{id,neg\}$, we
do not loose any generality. Indeed, if at least one of the nodes of the
circuit, say node $i$, has a constant local function then its incoming arc is
useless. Node $i$ does not depend on node $i-1$ and we no longer are looking at
a ``real'' circuit. Note that arcs $(i,i+1)$ such that $f_{i+1}=id$ ({\it resp.}
$f_{i+1}=neg$) are positive ({\it resp.} {\bf negative}).  A Boolean automata
circuit $C$ and the circuit associated, $\mathbb{C}_n=(V,A)$, are said to be
{\bf positive} ({\it resp.} {\bf negative}) if the number of negative arcs of
$A$ is even ({\it resp.} odd).  \medskip

Let $C=(\mathbb{C}_n,\{f_i\ |\ i\in V\})$ be a Boolean automata circuit of size $n$. In the
sequel, we will make substantial use of the following function:
$$
F[j,i]=
\left\lbrace  
  \begin{tabular}{ll}
    $f_j\circ f_{j-1}\circ\ldots\circ f_i$ & \hskip10pt if $i \leq j$ \\
    $f_j\circ f_{j-1}\circ\ldots\circ f_0\circ f_{n-1}\circ\ldots\circ f_{i}$  
    & \hskip10pt if $j < i$ \\
  \end{tabular}
\right.
$$
Because $\forall k,\ f_k\in\{id,neg\}$, $F[j,i]$ is injective. Also, note that
if $C$ is positive then $\forall j,\ F[j+1,j]=id$ and if, on the contrary, $C$
is negative then $\forall j,\ F[j+1,j]=neg$.

\subsection*{General iteration graphs of Boolean automata circuits}
In this section, we enumerate some preliminary results and eventually drive from
them a description of the general iteration graph of any Boolean automata
circuit.
\medskip

We will use the following notations: $\forall a\in\{0,1\},\ \neg a=neg(a)=1-a$
and $\forall x=(x_0,\ldots,x_{n-1})\in\{0,1\}^n,\ \overline{x}=(\neg
x_0,\ldots,\neg x_{n-1})$ (not to be confused with the complementary
$\overline{A}=\{a\notin A\}$ of a set $A$).  \medskip

The first result of this section will allow us to focus later on just one
instance of all Boolean automata circuit of given sign and size.
\begin{lemma}
  \label{iso}
  If $C$ and $C'$ are two Boolean automata circuits of same sign and size, then
  their general iteration graphs ${\cal D}(C)$ and ${\cal D}(C')$ are
  isomorphic.
\end{lemma}
\begin{proof}
  Suppose $C=(\mathbb{C}_n,{\cal F}=\{f_i\})$ and $C'=(\mathbb{C}_n,{\cal
    G}=\{g_i\})$ are two Boolean automata circuits of same sign and size $n$. We
  define the function $\sigma:\{0,1\}^n\to \{0,1\}^n$ satisfying the following:
  $$
  \sigma_i(x)\ =
  \begin{cases}
    x_i & \text{if } G[0,i+1]=id\\
    \neg x_i & \text{if } G[0,i+1]=neg.
  \end{cases}
  $$
  Let $x\in\{0,1\}^n$, $P\subseteq V=\{0,\ldots,n-1\}$, $y=G^P(\sigma(x))$ and
  $z=\sigma(F^P(x))$.  Then by definition of $y$ and $z$, $\forall i\in V$, the
  following holds:
  $$
  y_i\ =
  \begin{cases} 
    \sigma_i(x) & =
    \begin{cases} 
      x_i & \text{if }   (A)\ \Leftrightarrow\ i\notin P\text{ and }G[0,i+1]=id\\[3pt]
      \neg x_i & \text{if } (B)\ \Leftrightarrow\ i\notin P\text{ and }G[0,i+1]=neg 
    \end{cases}\\[5pt]
    g_i(\sigma_{i-1}(x))  & = 
    \begin{cases} 
      x_{i-1} & \text{if }(C)\ \Leftrightarrow\  i\in P \text{ and } 
      \begin{cases} 
        g_i=id \text{ and } G[0,i]=id\\   
        g_i=neg \text{ and }  G[0,i]=neg 
      \end{cases}\\
      
      \neg x_{i-1} & \text{if } (D)\ \Leftrightarrow\ i\in P\text{ and }
      \begin{cases}
        g_i=neg \text{ and } G[0,i]=id\\
        g_i=id \text{ and } G[0,i]=neg 
      \end{cases}
    \end{cases}
  \end{cases}
  $$
  and
  $$
  z_i\ =
  \begin{cases} 
    F^P(x) & =
    \begin{cases} 
      x_i & \text{if } (A)\\   
      x_{i-1} & \text{if }(E)\ \Leftrightarrow\ G[0,i+1]=id\text{ and }i\in P 
    \end{cases}\\
    \neg F^P(x) & =
    \begin{cases} 
      \neg x_i & \text{if }(B)\\ 
      \neg x_{i-1} & \text{if }(H)\ \Leftrightarrow\ G[0,i+1]=neg\text{ and }i\in P  
    \end{cases}
  \end{cases}
  $$
  Because $G[0,i+1]=G[0,i+1]\circ g_i\circ g_i= G[0,i]\circ g_i$, it holds
  that $(C)\ \Leftrightarrow\ [i\in P \text{ and } G[0,i+1]=id]\
  \Leftrightarrow\ (E)$ and for similar reasons $(D)\ \Leftrightarrow\ (H)$. As
  a result, $y=z$ and so does Lemma~\ref{iso}.
 \hfill$\Box$
\end{proof}
Following Lemma~\ref{iso} we define the representative of all positive Boolean
automata circuits of size $n$ as the circuit that has no negative arcs ($\forall
i<n,\ f_i=id$). We call this circuit the {\bf canonical positive circuit}. The
canonical negative circuit of size $n$ may be defined as the circuit
$C=(\mathbb{C}_n,{\cal F}=\{f_i=id\ |\ 0<i<n\}\cup\{f_0=neg\})$. However, in
most of the proofs that follow, for the sake of simplicity and because the
negative case is very similar to the positive case, we concentrate on positive
Boolean automata circuits.  The two following result appear in~\cite{Remy}. They
describe some useful properties of the function $U(\cdot)$. The first one of
them is infered from the fact that if $y=F(x)$, then $ i\in U(x)\Leftrightarrow
x_i\neq f_i(x_{i-1})=y_i \Leftrightarrow f_{i+1}(x_i)= y_{i+1}\neq f_{i+1}(y_i)
\Leftrightarrow i+1\in U(y)$.
\begin{lemma}
  \label{u_par}
  Let $C=(\mathbb{C}_n,{\cal F})$ be a Boolean automata circuit of size $n$.
  Then, $\forall x\in \{0,1\}^n,\ U(F(x))=\{i\in V\ |\ i-1\in U(x)\}$ so that
  $\forall k\in\mathbb{N},\ u(x)=u(F(x))=u(F^k(x))$.
\end{lemma}
\begin{lemma}
  \label{evenodd}
  Let $C$ be a Boolean automata circuit of size $n$. If $C$ is positive then
  $\forall x\in \{0,1\}^n$, $u(x)$ is even and if $C$ is negative then $\forall
  x\in \{0,1\}^n$, $u(x)$ is odd.
\end{lemma}
\begin{proof}
  Let $u=u(F(x))=u(F^k(x)),\ \forall k\in\mathbb{N}$.  Then, for any $i\in V$,
  $$
  u=|\{k\leq n\ |\ i\in U(F^k(x))\}|.
  $$
  \vspace{-4pt}
  And since
  \vspace{-4pt}
  $$
  F^n(x)_i=F[i,i+1](x_i)=
  \begin{cases}
    x_i & \text{ if } F[i,i+1]=id \text{ {\it i.e.,} if }C\text{ is positive,}\\
    \neg x_i & \text{ if } F[i,i+1]=neg \text{ {\it i.e.,} if }C\text{ is negative,}\\
  \end{cases}
 $$
  when $F$ is applied iteratively $n$ times, every node $i\in V$ must change
  states an even (resp. odd) number $u$ of times if $C$ is positive
  (resp. negative).\hfill$\Box$
\end{proof}
Focusing on the canonical positive Boolean automata circuit of size $n$ and on
its configurations of the form $X=(10)^k0^{n-2k}$ which clearly satisfy
$u(X)=2\cdot k$, we derive (extending the result to the negative case and using
Lemma~\ref{iso}) that:
\begin{multline*}
  u_{max}=max\{u(x)\ |\ x\in\{0,1\}^n\}\\= 
  \begin{cases}n & \text{ if }C\text{ is positive ({\it resp.} negative) and
    }n\text{ even ({\it resp.} odd), }\\
    n-1 & \text{ if }C\text{ is positive ({\it resp.} negative) and }n\text{ odd
      ({\it resp.} even).}
  \end{cases}
\end{multline*}
As for $u_{min}=min\{u(x)\ |\ x\in\{0,1\}^n\}$, it clearly equals $0$ when $C$ is
positive and $1$ when it is negative.
\medskip.

$N$ being a Boolean automata network of size $n$, we define the following binary
relations between configurations $x,y\in\{0,1\}^n$ of $N$:
$$
x\ \rightarrow y\ \Leftrightarrow\ \exists P\subseteq V,\ y=F^P(x).
$$
$$
x\stackrel{\ast}{\leftrightarrow}y\ \Leftrightarrow\
x\stackrel{\ast}{\rightarrow}y\text{ and }y\stackrel{\ast}{\rightarrow}x
$$
where $\stackrel{\ast}{\rightarrow}$ is the reflexive and transitive closure of
$\rightarrow$. Thus, $x\ \rightarrow y$ is equivalent to the existence of the
arc $(x,y)$ in ${\cal D}(N)$ and $x\stackrel{\ast}{\rightarrow}y$ is equivalent
to there being an oriented path in ${\cal D}(N)$ from $x$ to $y$. The two
following lemmas characterise the relations $\stackrel{\ast}{\rightarrow}$ and
$\stackrel{\ast}{\leftrightarrow}$.
\begin{lemma}
  \label{horiz}
  Let $C$ be a Boolean automata circuit of size $n$. Then, $\forall x,y\in
  \{0,1\}^n,$
  $$
  u(x)=u(y)\ \Rightarrow\
  x\stackrel{\ast}{\leftrightarrow} y.
  $$
\end{lemma}
\begin{proof}
  We prove Lemma~\ref{horiz} in the case where $C$ is the canonical Boolean
  automata circuit of size $n$.  First note that because $\forall i,\ f_i=id$,
  it holds that $\forall i\in U(x),\ x_i\neq x_{i-1}$.  Note also that
  $F(x_0,\ldots,x_{n-2},x_{n-1})=(x_{n-1},x_0\ldots,x_{n-2})$. As a result, for
  any configuration $x$ such that $u(x)=k$, there exists integers
  $a_1,\ldots,a_k,b_1,\ldots,b_k$ such that:
  \begin{equation}
    \label{form}
    x\stackrel{\ast}{\leftrightarrow} x'=1^{a_1}0^{b_1}1^{a_2}0^{b_2}\ldots
    1^{a_k}0^{b_k}.
  \end{equation}
  Indeed, if $x$ is not already in the form of the configuration $x'$ above,
  then $\exists a\in \mathbb{N},\ x=x_0\ldots x_{n-2-a}01^a$ and/or $\exists
  b\in \mathbb{N},\ x=0^b1x_{b}\ldots x_{n-1}$. In the first case,
  $x'=F^a(x)=1^ax_0\ldots x_{n-2-a}0$ and in the second $x'=F^{n-b}(x)=1
  x_{b}\ldots x_{n-1}0^b$.  Now, define $X=(10)^{k}0^{n-2k}$. We claim that for
  any $x$ in the form of $x'$ in~(\ref{form}), $x
  \stackrel{\ast}{\leftrightarrow} X$ so that, as a consequence, $\forall
  x,y\in\{0,1\}^n,\ u(x)=u(y)=k\ \Rightarrow\ x\stackrel{\ast}{\leftrightarrow}
  X\stackrel{\ast}{\leftrightarrow}y.$ Define the $k$-uples
  $n(x)=(n-b_k,a_1,\ldots,a_k)$ and order them lexicographically.  If
  $x=1^{a_1}0^{b_1}1^{a_2}0^{b_2}\ldots 1^{a_k}0^{b_k}$ is such that
  $n(x)>n(X)=(2k-1,1,\ldots,1)$, then there exists
  $x'=1^{a'_1}0^{b'_1}1^{a'_2}0^{b'_2}\ldots 1^{a'_k}0^{b'_k}$ such that
  $n(x')<n(x)$. Indeed, first suppose that $\exists i\leq k,\ a_i>1$, {\it
    i.e.,} $ x=1^{a_1}0^{b_1}\ldots 0^{b_{i-1}}11^{a_i-1}0^{b_i}
  \ldots1^{a_k}0^{b_k}$. Let $j=\sum_{l=1}^{i-1}a_l+b_l$ be the first node
  of the $i^{th}$ section of $1$s in $x$ and let
  $$
  x'=F^{\{j\}}(x)= 1^{a_1}0^{b_1}\ldots
  0^{b_{i-1}}0 1^{a_i-1}0^{b_i} \ldots 1^{a_k}0^{b_k}.
  $$ 
  Then, $n(x')=(n-b_k,a_1,\ldots,a_i-1,a_i-1, a_{i+1},\ldots,a_k)<n(x)$.  Now,
  suppose that $\forall i\leq k,\ a_i=1$: $x=10^{b_1}1 0^{b_{2}}1 \ldots
  0^{b_{k-1}}10^{b_k}$.  Let $j=\sum_{l=1}^{k-1}1+b_l=n-b_k$ be the first and
  only node of the $k^{th}$ section of $1$s in $x$ and let
  $$
  \begin{tabular}{l}
    $x'''= F^{\{n-2 =j+b_k-1\}} \circ\ldots\circ F^{\{j+1 \}}(x)= x_0\ldots
    x_{j-b_{k-1}-1}0^{b_{k-1}} 1 1^{b_k-1}0$,\\[3pt]
    $x''= F^{\{n-3 \}} \circ\ldots \circ F^{\{j+1\}}\circ
    F^{\{j\}}(x''')=x_0\ldots x_{j-b_{k-1}-1}0^{b_{k-1}} 0^{b_k-1}10$ \text{ and
    }\\[3pt]
    $x'=F^2(x'')=10x_0\ldots x_{j-b_{k-1}-1}0^{b_{k-1}+b_k-1}.$
  \end{tabular}
  $$
  We have supposed above that $b_{k-1}>1$. If it isn't, we can always consider
  the configuration $F^{b_k+1}(x)$ instead of $x$.  Then,
  $n(x')=(n-b_k-b_{k-1}+1,1,\ldots,1)<n(x)=(n-b_k,1,\ldots,1)$. By induction on
  $n(x)$, we derive from this that $x\stackrel{\ast}{\rightarrow} X$. In a
  similar manner we show that $X\stackrel{\ast}{\rightarrow} x$ and we are done.
  \hfill$\Box$
\end{proof}
\begin{lemma}
  \label{down}
  Let $C$ be a Boolean automata circuit of size $n$. Then $ \forall x,y\in
  \{0,1\}^n,$
  $$
  u(x)>u(y)\ \Rightarrow\ [x\stackrel{\ast}{\rightarrow} y]\ \vee \neg
  [y\stackrel{\ast}{\rightarrow} x].  
  $$
\end{lemma}
\begin{proof}
  Again, let us prove Lemma~\ref{down} in the case where $C$ is the canonical
  positive circuit.  Considering Lemma~\ref{horiz}, the first part of
  Lemma~\ref{down}, when $u(y)>0$, comes from:
  $$
  X=(10)^{k}.10.0^{n-2k-2}\stackrel{\ast}{\rightarrow}Y=F^{\{2k\}}(X)=
  (10)^{k}0^{n-2k}
  $$
  where $u(X)=2\cdot (k+1)$ and $u(Y)=2\cdot k$.  If $u(y)=0$, let
  $X=10^{n-1}$. Then, necessarily, $y=0^{n}=F^{\{0\}}(X)$ or
  $y=1^n=F^{\{n-1\}}\circ \ldots\circ F^{\{2\}}\circ F^{\{1\}}(X)$.  Now,
  suppose that the second part of Lemma~\ref{down} is false. Again, from
  Lemma~\ref{horiz}, this means that
  \begin{equation}
    \label{up}
    Y=(10)^{k}.00.0^{n-2k-2k}\stackrel{\ast}{\rightarrow}X=(10)^{k}.10.0^{n-2k-2}.
  \end{equation}
  But as it is easy to see that although the sizes of sections of consecutive
  $1$s can be increased or decreased ``at will'', these sections cannot be made
  to appear inside a section of consecutive $0$s. Equation~\ref{up} is therefore
  impossible.\hfill$\Box$
\end{proof}
Following Lemma~\ref{down}, we may notice that $u(\cdot)$ serves as an obvious
potential function. Informally, the most {\em robust} and {\em likely} (in the
sense mentioned in section~\ref{appl}) or most {\em stable} configurations $x$
are those of lesser potential $u(x)$.
\begin{lemma}
  \label{ck}
  Let $C$ be a Boolean automata circuit of size $n$ and let $U_k=\{x\in
  \{0,1\}^n\ |\ u(x)=k\}$. Then, for any $k\leq u_{max},$
  $$
  u_k\ =\ |U_k|\ =\ 2\times \binom{n}{k}.
  $$
\end{lemma}
\begin{proof} 
  A configuration $x$ such that $u(x)=k$ is completely defined by the set
  $U(x)=\{i_1,\ldots,i_{k}\}$ and by the value of the state of one node, say
  nodes $i_1$ (indeed, $\forall i\notin U(x), x_i=f_{i}(x_{i-1})$, and $\forall
  i\in U(x), x_i\neq f_{i}(x_{i-1})$). Lemma~\ref{ck} comes from the fact that
  there are $\binom{n}{k}$ ways to choose the $k$ nodes of $U(x)$ and $2$ ways
  to choose the state of node $i_1$ (either $1$ or $0$). \hfill$\Box$
\end{proof}
Note that $u(x)=0$ is equivalent to $x$ being a fixed point and as it is well
known and easy to see that a positive circuit has two fixed points $x$ and
$\overline{x}$ ($x=0^n$ and $\overline{x}=1^n$ in the case of the canonical
positive circuit). This confirms $u_0=2\times \binom{n}{0}=2$. In the negative
case, according to Lemma~\ref{ck}, the set of configurations of minimal
$u(\cdot)$ is of size $2\times \binom{n}{1}=2\cdot n$.
In the general case, following Lemma~\ref{ck}, we may notice, as did the
authors of~\cite{Remy}, that if $x\in U_k$, then $\overline{x}\in U_k$.
\medskip

Resulting from the previous remarks and lemmas, the general iteration graph
${\cal D}(C)$ representing the dynamics of any Boolean automata circuit $C$ can
be represented by a layered graph as in figure~\ref{couches}. Layer of $k\leq
u_{max}$ of this graph contains all configurations of $U_k$.  \medskip
\hspace{-10cm}~\begin{figure}[htbp!]
  \begin{center}
    \hspace{-3cm}\scalebox{0.5}{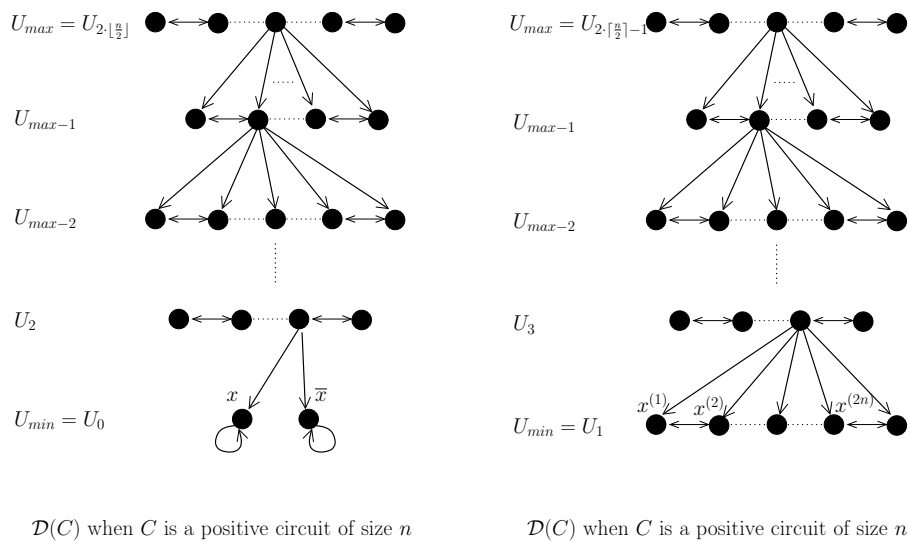}
    \caption{{\small General iteration graphs ${\cal D}(C)$ of an arbitrary
        positive Boolean automata circuit $C$ of size $n$ (on the left) and of
        an arbitrary negative Boolean automata circuit $C$ of size $n$ (on the
        right).}}
  \end{center}
  \label{couches}
\end{figure}

In~\cite{Remy}, Remy {\it et al.} concentrated on synchronous and asynchronous
updates. They called, respectively, ``synchronous graph'' and ``asynchronous
graph'' the iteration graphs of circuits with these updates. Both these graphs
are sub-graphs of the general iteration graph of a circuit.  More precisely,
arcs of a synchronous graph are the arcs of the general iteration graph that are
labeled by $V$ and arcs of the asynchronous graphs are the arcs of the general
iteration graphs that are labeled by a subset $P\subseteq V$ containing just one
node ($|P|=1$). As we did here for general iteration graphs, Remy {\it et al.}
showed that synchronous and asynchronous graphs are also layered graphs whose
layers are characterised by the value of $u(\cdot)$.
Let us seize the opportunity here to mention with no further detail that
in~\cite{Pcircuits}, we described exhaustively the synchronous graph by focusing
on and characterising attractors and in~\cite{Remy} the authors do the same
thing by different means exploiting the sets $U_k,\ k\leq u_{max}$. The
descriptions produced by these two analyses of the dynamics of circuits updated
in parallel are necessarily different but combining them yields some
supplementary precisions concerning the sets $U(x)$ for configurations $x$
belonging to attractors of Boolean automata circuits updated synchronously.
\medskip

In~\cite{Thomas1981}, Thomas formulated two conjectures establishing the
relationship thought to exist between the structure of a network and its
dynamical behaviour. The first of these conjectures claiming that a network
needs to contain positive circuits in order for it to have fixed points, is now
a well known fact that has been proven true in many
contexts~\cite{conjT1,conjT2,conjT3,conjT4,conjT5,RRT,Thomas07} (it is besides
confirmed again by the fact that there are no $u(x)=0$ potential configurations
of a negative circuit). The second conjecture, on the other hand, states that
negative circuits are necessary to have limit cycles. This also has been proven
in some contexts~\cite{conjT1,conjT2,conjT3,conjT6, AR1,AR2,AR3,RRT,Thomas07}
However, in our context of Boolean automata networks, the results we found
in~\cite{Pcircuits} and in~\cite{BScircuits} have proven it to be false. Indeed,
we have shown that positive Boolean automata networks updated with
block-sequential update schedules do have limit cycles as long as the update
schedule is not one of the $n$ sequential update schedules $s$ such that
$\exists i\leq n,\ s(i)=0, s(i+1)=1,\ldots,s(i+n-1)=n-1$. This inconsistency of
our model with those in agreement with the second Thomas conjecture disappears
if we consider, as we did earlier, that it is highly improbable that no errors
ever occur in the fixed update order. If we do, then, as the general iteration
graphs of circuits signify, fixed points having the least potential are almost
the only possible outcome of the evolution of an isolated positive circuit.

\section{Discussion}
Assuming that neither total asynchrony nor perfect obedience to a
block-sequential update mode are probable, we believe that general iteration
graphs convey more realistic information on a network dynamics than do simple
iteration graphs of particular block-sequential update schedules. Indeed, as we
have endeavoured to show through the analyses of the Mendoza network in
Section~\ref{appl} and of arbitrary circuits in Section~\ref{circ}, they reveal
that some behaviours of networks that seam dynamically stable with a specific
block-sequential update mode, become very less stable when the update mode may
occasionally be disrupted. This is particularly clear with Boolean automata
circuits for which, as we have seen above, there exists a rather straightforward
potential function of network configurations. For more elaborate networks such
as the Mendoza network, it is less obvious.  From the functions of robustness and
likeliness that we have defined in Section~\ref{appl}, however unsubtle may they
be, we have still managed to derive that some attractors that are made possible
by a given block-sequential update schedule may be much less probable and thus
much less meaningful from a biological point of view than others.  We wish,
however, to develop and refine these notions of robustness and likeliness in the
hope of extending the potential function of circuits to other networks.

General iteration graphs have one major drawback: their sizes. They have $2^n$
nodes and $2^n\times(2^n-1)$ arcs (where $n$ is the network size), that is,
$2^n-1$ arcs more than iteration graphs of block-sequential update
schedules. This problem has yet to be dealt with before we may start considering
to compute the general iteration graphs of arbitrary networks whose sizes may be
much bigger than, for instance, the ones of the two strongly connected
components of the Mendoza network that we have studied in
Section~\ref{appl}. However, like Thomas~\cite{Thomas1981}, we believe that
understanding exhaustively the dynamics of circuits is a step towards building
an understanding of that of arbitrary networks. The knowledge of the general
iteration graphs of circuits that we now have may allow us eventually to bypass
the costly construction of the general iteration graphs of other networks by
focusing on these motifs in the networks structures.

\bibliography{GIG}

\begin{thebibliography}{10}

\bibitem{Julio_rob}
J.~Aracena, E.~Goles, A.~Moreira, and L.~Salinas.
\newblock On the robustness of update schedules in boolean networks.
\newblock {\em Biosystems}, 97, 2009.

\bibitem{conjT5}
O.~Cinquin and J.~Demongeot.
\newblock Positive and negative feedback: striking a balance between necessary
  antagonists.
\newblock {\em Journal of Theoretical Biology}, 216:229--241, 2002.

\bibitem{adrien}
J.~Demongeot, A.~Elena, and S.~Sen\'e.
\newblock Robustness in regulatory networks: a multi-disciplinary approach.
\newblock {\em Acta Biotheoretica}, 56(1-2):27--49, 2008.

\bibitem{Pcircuits}
J.~Demongeot, M.~Noual, and S.~Sen\'e.
\newblock On the number of attractors of boolean automata circuits.
\newblock {\em BLSMC, in press}, 2010.

\bibitem{these_adrien}
A.~Elena.
\newblock {\em Robustesse des r{\'e}seaux d'automates bool{\'e}ens {\`a} seuil
  aux modes d'it{\'e}ration. Application {\`a} la mod{\'e}lisation des
  r{\'e}seaux de r{\'e}gulation g{\'e}n{\'e}tique}.
\newblock PhD thesis, Universit{\'e} Joseph Fourier - Grenoble, 2009.

\bibitem{BScircuits}
E.~Goles and Noual.
\newblock Block-sequential update schedules and boolean automata circuits.
\newblock 2009.

\bibitem{comp}
E.~Goles and L.~Salinas.
\newblock Comparison between parallel and serial dynamics of boolean networks.
\newblock {\em Theoretical Computer Science}, (1--3):247--253, 2008.

\bibitem{conjT2}
J.-L. Gouz\'e.
\newblock Positive and negative circuits in dynamical systems.
\newblock {\em Journal of Biological Systems,}, 6:11--15, 1998.

\bibitem{conjT6}
M.~Kaufman, C.~Soul\'e, and R.~Thomas.
\newblock A new necessary condi- tion on interaction graphs for
  multistationarity.
\newblock {\em Journal of Theoretical Biology}, 248:675--685, 2007.

\bibitem{Thomas07}
M.~Kaufman, C.~Soul\'e, and R.~Thomas.
\newblock A new necessary condition on interaction graphs for
  multistationarity.
\newblock {\em Journal of theoretical Biology}, 2007.

\bibitem{mend}
L.~Mendoza and E.~R. Alvarez-Buylla.
\newblock Dynamics of the genetic regulatory network for arabidopsis thaliana
  flower morphogenesis.
\newblock {\em Journal of Theoretical Biology}, 193:307--319, 1998.

\bibitem{conjT1}
E.~Plathe, T.~Mestl, and S.~W. Omholt.
\newblock Feedback loops, stability and multistationarity in dynamical systems.
\newblock {\em Journal of Biological Systems}, 3:569--577, 1995.

\bibitem{Remy}
E~Remy, B.~Moss\'e, C.~Chaouiya, and D.~Thieffry.
\newblock A description of dynamical graphs associated to elementary regulatory
  circuits.
\newblock {\em Bioinformatics}, 19(2):172--178, 2003.

\bibitem{RRT}
E.~Remy, P.~Ruet, and D.~Thieffry.
\newblock Graphic requirements for multistability and attractive cycles in a
  boolean dynamical framework.
\newblock {\em Advances in applied mathematics}, 41(3):335--350, 2008.

\bibitem{AR1}
A.~Richard.
\newblock On the link between oscillations and negative circuits in discrete
  genetic regulatory networks.
\newblock JOBIM, 2007.

\bibitem{AR2}
A.~Richard.
\newblock Negative circuits and sustained oscillations in asynchronous automata
  networks.
\newblock {\em Advances in Applied Mathematics}, 2009.

\bibitem{AR3}
A.~Richard and J.-P. Comet.
\newblock Necessary conditions for multistationarity in discrete dynamical
  systems.
\newblock {\em Discrete Applied Mathematics}, 155(18):2403--2413, 2007.

\bibitem{theseS}
S.~Sen\'e.
\newblock {\em Influence des conditions de bord dans les r\'eseaux d'automates
  bool\'eens \`a seuil et application \`a la biologie}.
\newblock PhD thesis, Universit\'e Joseph Fourier de Grenoble, 2008.

\bibitem{conjT3}
E.~H. Snoussi.
\newblock Necessary conditions for multistationarity and stable periodicity.
\newblock {\em Journal of Biological Systems}, 6:3--9, 1998.

\bibitem{conjT4}
C.~Soul\'e.
\newblock Mathematical approaches to differentiation and gene regulation.
\newblock {\em Comptes rendus de l'Acad\'emie des sciences, Biologies},
  329:13--20, 2006.

\bibitem{Thomas1981}
R.~Thomas.
\newblock On the relation between the logical structure of systems and their
  ability to generate multiple steady states or sustained oscillations.
\newblock {\em Springer Series in Synergetics}, 9, 1981.

\end{thebibliography}
\end{document}